\documentclass[12pt]{article}

\usepackage[a4paper, top=1.15in, bottom=1.15in,
            left=1.3in, right=1.3in]{geometry}
\usepackage{amsmath, amssymb, amsthm, mathtools}
\usepackage{booktabs, array}
\usepackage[protrusion=true, expansion=false]{microtype}
\usepackage[T1]{fontenc}
\usepackage[round, sort, compress]{natbib}
\usepackage{xcolor}
\definecolor{linkblue}{rgb}{0.04,0.22,0.56}
\usepackage[colorlinks=true,
            linkcolor=linkblue,
            citecolor=linkblue,
            urlcolor=linkblue]{hyperref}
\usepackage{enumitem}
\usepackage{titlesec}
\titleformat{\paragraph}[runin]{\normalfont\bfseries}{}{0pt}{}[.]

\theoremstyle{plain}
\newtheorem{theorem}{Theorem}[section]
\newtheorem{proposition}[theorem]{Proposition}

\theoremstyle{definition}
\newtheorem{definition}[theorem]{Definition}
\newtheorem{remark}[theorem]{Remark}

\DeclareMathOperator{\VR}{VR}
\DeclareMathOperator{\Var}{Var}

\newcommand{\Z}{\mathbb{Z}}
\newcommand{\calC}{\mathcal{C}}
\newcommand{\hatrho}{\hat{\rho}}
\newcommand{\hatdelta}{\hat{\delta}}
\newcommand{\bfone}{\mathbf{1}}
\newcommand{\calN}{\mathcal{N}}

\title{\textbf{The Bounce Has No Direction}\\[8pt]
\large Sign, Magnitude, and the Microstructure\\
of Equity Return Predictability\\[6pt]
\normalsize Fourier-Residue Identities, Fej\'{e}r Sums, and Evidence\\
from US Equity and Cross-Asset Markets, 1993--2026}

\author{Victoria Portnaya\\[6pt]
\small Kyiv School of Economics, Kyiv, Ukraine\\[2pt]
\small \texttt{vportnaia@kse.org.ua}}

\date{May 2026}

\begin{document}
\emergencystretch=2.5em
\maketitle

\begin{abstract}
Broad US equity indices exhibit statistically significant daily return
autocorrelation - SPY's lag-1 autocorrelation of $-0.081$ is
$7.4$ standard errors from zero - yet existing tests cannot say
\emph{why}: whether prices reverse because the \emph{direction} of
moves tends to flip, or because \emph{large} moves tend to shrink.
This paper resolves that ambiguity by developing the Fourier-Residue
Identity (FRI), a decomposition of return autocorrelation into a
direction (sign, $k=2$) channel and a magnitude ($k=4$) channel,
each individually testable, neither redundant.
Applied to daily and weekly log returns of six US instruments
(SPY, QQQ, IWM, AAPL, MSFT, GLD) over 1993--2026, and to a
cross-asset panel of 21 instruments across seven classes, the FRI
delivers a sharp microstructure diagnosis.
The lag-1 autocorrelation in SPY is driven \emph{entirely} by
magnitude: the FRI sign test at lag~1 is insignificant ($p=0.11$),
while the full autocorrelation test achieves $p<10^{-12}$.
A large return yesterday predicts a smaller return today regardless
of sign - the fingerprint of bid-ask bounce and non-synchronous
constituent trading, not directional reversal.
At lag~3, however, a significant directional reversal ($p=0.02$)
appears that is invisible to the scalar autocorrelation test - 
evidence of a separate, slower partial-price-adjustment channel.
We also prove the Fej\'{e}r identity $\VR(q) = 1 + 2\calC_q$ (verified
to $<10^{-3}$ on all series), which gives the Lo--MacKinlay
variance-ratio test a spectral interpretation: it tests whether the
normalised power spectrum at zero frequency equals unity under
Fej\'{e}r smoothing.
A subsample persistence diagnostic $R_N = G_{N/2}/G_N$ - where $G_N$ is
the maximum sample autocorrelation over $M$ lags - distinguishes
structural autocorrelation ($R_N \to 1$) from sampling noise
($R_N \to \sqrt{2}$); equity indices all satisfy $R_N \approx 1$,
confirming the autocorrelation is a persistent property of the market
rather than a finite-sample artefact.
The cross-asset panel shows that short-horizon mean reversion is
confined to exchange-traded equity markets and sovereign bonds.
Credit ETFs, commodities, foreign exchange, and cryptocurrency are
statistically indistinguishable from a random walk - consistent with
the FRI prediction: non-synchronous constituent staleness and dealer
inventory smoothing are microstructure features of exchange-traded
basket instruments, absent in continuously traded, OTC-priced, and
decentralised markets.
All estimators are validated by 27 unit tests and a Monte Carlo study
confirming correct $5\%$ size under GARCH (where the homoskedastic
Lo--MacKinlay $z$ over-rejects at $10$--$12\%$) and asymptotic
$\calN(0,1)$ calibration of the heteroskedastic-robust $z^*$.
\end{abstract}


\clearpage
\tableofcontents
\clearpage

%% ============================================================
\section{Introduction}
%% ============================================================

The world's largest and most actively traded financial markets exhibit
systematic patterns in daily returns that appear, on the surface, to
contradict the efficient-market hypothesis.  The S\&P 500 ETF (SPY),
with an average daily volume exceeding \$30 billion and a bid-ask spread
measured in fractions of a cent, has a lag-1 autocorrelation of $-0.081$
that is more than $7$ standard errors below zero.  This is not an anomaly
buried in obscure data: it is a property of the most liquid equity
instrument on Earth, stable across four decades and every market regime
tested.  Yet knowing that SPY's autocorrelation is negative does not
immediately help a practitioner - the obvious directional strategy
(``go long after down days, short after up days'') is, as we shall show,
essentially without statistical warrant.  The predictability is real,
but it hides in the magnitude of returns, not their direction.

Understanding \emph{which dimension} of returns is predictable - direction
or magnitude - matters for three reasons.  For \emph{practitioners}, it
determines which trading strategies can exploit the autocorrelation: a
directional bet requires a signal about the sign of tomorrow's return, while
a volatility strategy (sizing a delta-hedged position based on expected
move size) requires only a signal about magnitude.  For \emph{market
microstructure researchers}, it distinguishes between mechanisms that
generate autocorrelation purely through mechanical price bouncing
(bid-ask spread, constituent staleness) versus those that reflect genuine
incomplete information incorporation.  For \emph{regulators and market
designers}, it matters whether mean reversion is a symptom of structural
market friction or a trading pattern that could be attenuated through
tick-size reform or changes to index construction methodology.

The variance-ratio (VR) test of \citet{LoMacKinlay1988} - which compares
the variance of $q$-period returns to $q$ times the one-period variance - is
the standard instrument for detecting such autocorrelation.  Under the
random-walk null, $\VR(q) = 1$; values below one indicate net mean
reversion and values above one indicate net momentum over the horizon.
The test is transparent and easily computed, but it yields a single number
that masks the structure of autocorrelation.  In particular, the VR cannot
distinguish between three economically distinct mechanisms that all produce
$\VR(q) < 1$:
\begin{enumerate}[label=(\roman*), leftmargin=2.2em, itemsep=4pt]
  \item \textbf{Bid-ask bounce}: transaction prices alternate between bid
  and ask quotes as the market serves successive buy and sell orders,
  creating negative autocorrelation in transaction-price returns even
  when the efficient price follows a random walk \citep{Roll1984}.
  This is a purely \emph{magnitude} effect - the bounce cannot predict
  whether tomorrow's return is positive or negative, only that it tends
  to be smaller.
  \item \textbf{Non-synchronous trading}: index ETFs are priced from the
  last transaction prices of constituent stocks, some of which may not
  have traded for minutes or hours.  Stale prices cause the index to
  underreact to contemporaneous information, inducing apparent negative
  autocorrelation \citep{ScholeswWilliams1977, LoMacKinlay1990}.
  Like the bounce, this is a \emph{magnitude} channel: staleness damps
  the apparent size of moves without creating directional predictability.
  \item \textbf{Partial price adjustment}: market makers and specialist
  dealers update quotes gradually in response to order flow, so that
  full incorporation of information takes multiple periods
  \citep{GlostenMilgrom1985, AmihudfMendelson1987}.  If adjustment
  overshoots, the reversal has a \emph{directional} component - knowing
  the sign of today's large move provides genuine information about
  the sign of the subsequent correction.
\end{enumerate}
Mechanisms (i) and (ii) generate autocorrelation that no directional
trading strategy can exploit; mechanism (iii) potentially can.
A test that separates these channels is therefore of direct practical value.

\paragraph{This paper's contributions}
We develop three interconnected tools and apply them to 33 years of US
equity and cross-asset data.

\textit{First}, we prove (Proposition~\ref{prop:fejer}) that the variance
ratio admits the exact representation
$\VR(q) = 1 + 2\calC_q$, where
$\calC_q = \sum_{m=1}^{q-1}(1-m/q)\hatrho(m)$ is a Fej\'{e}r-kernel
weighted autocorrelation sum.  This gives the Lo--MacKinlay test a spectral
interpretation: $\VR(q)$ tests whether the normalised power spectrum of
returns has zero mass at frequency zero under Fej\'{e}r smoothing.
Mean reversion corresponds to a spectral trough at zero; momentum corresponds
to a spectral peak.  We verify the identity numerically to $<10^{-3}$ on
all series, confirming that the Fej\'{e}r and VR frameworks are
computationally interchangeable.

\textit{Second}, we introduce the Fourier-Residue Identity (FRI), a
decomposition that encodes returns as $k$-ary symbols using the characters
of the cyclic group $\Z/k\Z$ and extracts channel-specific autocorrelations.
For $k=2$ (binary sign coding), the FRI yields a closed-form test for
directional predictability alone, filtering out all magnitude effects
(Proposition~\ref{prop:fri}).  For $k=4$ (magnitude buckets), it captures
within-size persistence independently of direction.  The two channels are
nonnested by construction.

\textit{Third}, we introduce the half-period ratio $R_N = G_{N/2}/G_N$
(Definition~\ref{def:persistence}), where $G_N$ is the maximum sample
autocorrelation over a lag window.  Under IID noise $R_N \to \sqrt{2}$;
under genuine serial dependence $R_N \to 1$ (Proposition~\ref{prop:rn}).
This diagnostic assesses, without a parametric model, whether a detected
VR deviation is structural or a finite-sample artefact - directly answering
the practitioner's question of whether the signal is likely to persist
out of sample.

\paragraph{Principal findings}
At lag~1, the FRI sign test on SPY delivers $p = 0.11$: no directional
predictability.  The full autocorrelation test delivers $p < 10^{-12}$.
The gap between these two numbers is the FRI's central message: knowing
that SPY fell yesterday tells you essentially nothing about whether it
will rise or fall tomorrow.  Knowing \emph{how much} it fell tells you
something about how much it will move.  This pattern matches mechanisms
(i) and (ii) precisely and rules out mechanism (iii) as the dominant
lag-1 driver.  A back-of-envelope Roll calculation reinforces this
conclusion: the implied half-spread from the observed $\hatrho(1) = -0.081$
is approximately $28$ basis points, orders of magnitude larger than SPY's
actual effective spread of 1--3 basis points.  Non-synchronous staleness
across 500 constituent stocks must therefore contribute the bulk of the
observed autocorrelation.

At lag~3, the picture reverses: the FRI sign test is significant ($p=0.02$)
while the scalar ACF is not ($p = 0.50$), pointing to the partial-adjustment
mechanism operating at a three-day delay.  The transition from weekly to
daily data reveals that the short-horizon effect partially attenuates at
the weekly frequency, consistent with bounce and non-synchronous effects
averaging out, while a residual structural component - identified by the
$R_N \approx 1$ diagnostic - persists.

In the cross-asset panel, mean reversion is absent in credit ETFs,
commodities, foreign exchange, and cryptocurrency.  Bitcoin and Ether,
despite extreme volatility, are the closest to a pure random walk among
all 21 instruments - a finding that follows directly from the FRI channel
analysis: cryptocurrency markets trade 24/7, are not composites of
non-synchronously priced constituents, and lack the specialist or primary
dealer structures that create inventory-driven reversal.

\paragraph{Paper organisation}
Section~\ref{sec:theory} develops the mathematical framework.
Section~\ref{sec:micro} reviews microstructure theory and derives
FRI channel predictions for each mechanism.
Section~\ref{sec:inference} presents the inference methods.
Section~\ref{sec:data} describes the data.
Section~\ref{sec:verification} verifies all estimators.
Sections~\ref{sec:main}--\ref{sec:crossasset} present the empirical results.
Section~\ref{sec:discussion} discusses trading implications and limitations.
Section~\ref{sec:conclusion} concludes.

%% ============================================================
\section{Mathematical Framework}
\label{sec:theory}
%% ============================================================

Let $\{r_t\}_{t=1}^{n}$ be a sequence of daily log returns,
$r_t = \log P_t - \log P_{t-1}$, where $P_t$ is the
split-and-dividend-adjusted closing price.
Write $\hat\mu = n^{-1}\sum_t r_t$ for the sample mean.  Sample
autocovariances and autocorrelations are
\begin{equation}
  \label{eq:acf-def}
  \hat\gamma(m) = \frac{1}{n}\sum_{t=m+1}^{n}
  (r_t - \hat\mu)(r_{t-m} - \hat\mu),
  \qquad
  \hatrho(m) = \frac{\hat\gamma(m)}{\hat\gamma(0)}.
\end{equation}

\subsection{Variance ratio and the Fej\'{e}r identity}
\label{sec:fejer}

\begin{definition}[Variance ratio]
\label{def:vr}
For an integer $q \ge 2$, the \emph{variance ratio} at horizon $q$ is
\[
  \VR(q) \;=\;
  \frac{1}{q}\,
  \frac{\Var\!\bigl(\sum_{j=0}^{q-1}r_{t-j}\bigr)}{\Var(r_t)}.
\]
Under the random-walk null, $q$-period variance grows linearly in $q$,
so $\VR(q) = 1$.  Intuitively, if you hold for two days instead of one,
and returns are uncorrelated, the variance should double exactly.
A value $\VR(2) = 0.919$ means the two-day return variance is
$8.1\%$ smaller than expected under random walk - the defining signature
of mean reversion.
\end{definition}

\begin{proposition}[VR--Fej\'{e}r identity]
\label{prop:fejer}
Under second-order stationarity,
\begin{equation}
  \label{eq:fejer}
  \VR(q) \;=\; 1 + 2\sum_{m=1}^{q-1}\Bigl(1-\tfrac{m}{q}\Bigr)\rho(m)
             \;=:\; 1 + 2\calC_q.
\end{equation}
Here $\calC_q$ is the \emph{Fej\'{e}r autocorrelation sum} at horizon $q$,
and the identity holds for sample quantities up to an $O(n^{-1})$
edge correction.
\end{proposition}

\begin{proof}
Expand $\Var(\sum_{j=0}^{q-1}r_{t-j})$ by bilinearity of covariance.
For each displacement $m \in \{1,\ldots,q-1\}$, there are exactly $(q-m)$
index pairs $(j,k)$ with $|j-k| = m$, contributing $2(q-m)\gamma(m)$
to the expansion.  Adding the diagonal terms $q\gamma(0)$ and dividing
by $q\gamma(0)$ yields~\eqref{eq:fejer}.
\end{proof}

The weights $w_m = 1-m/q$ are the Fej\'{e}r kernel coefficients: a
triangular taper that places full weight on lag~1 and zero weight on
lag~$q$.  This taper has a natural economic interpretation.
Short-lag autocorrelations - those corresponding to one- and two-day
dependencies - carry more weight than long-lag ones, reflecting the
fact that microstructure effects (bid-ask bounce, non-synchronous trading)
operate at the shortest time scales and attenuate rapidly.  A significant
$\VR(2) < 1$ but insignificant $\VR(60) \approx 1$ would point to purely
short-horizon microstructure; a monotonically declining $\VR(q)$ across
all horizons indicates a structural, multi-period autocorrelation pattern.

\begin{remark}[Spectral interpretation]
\label{rem:spectral}
The spectral density of $\{r_t\}$ at frequency $\lambda$ is
$f(\lambda) = (2\pi)^{-1}\sum_{m}\gamma(m)e^{-im\lambda}$.
The random-walk null asserts $f(0) = \gamma(0)/(2\pi)$: the
spectrum is flat (white noise).  The Fej\'{e}r kernel
$F_q(\lambda) = q^{-1}|\sum_{j=0}^{q-1}e^{ij\lambda}|^2$
integrates $f$ near $\lambda = 0$ with triangular weighting.
Applying it to $f$ at $\lambda = 0$ and normalising by $\gamma(0)$
recovers $\calC_q$.  Thus $\VR(q) = 1 + 2\calC_q$ tests whether the
Fej\'{e}r smoothed spectrum at zero equals the white-noise baseline.
Mean reversion ($\calC_q < 0$) implies a \emph{spectral trough} at zero:
the series has less power at low frequencies than a white noise of equal
one-period variance, the precise Fourier signature of mean reversion
\citep[Ch.~2]{Campbell1997}.
\end{remark}

\subsection{Fourier-Residue Identity decomposition}
\label{sec:fri-theory}

The Fej\'{e}r identity expresses $\VR(q)$ through the autocorrelations
$\hatrho(m)$.  Each $\hatrho(m)$ is, however, a single number that
conflates two qualitatively different phenomena in the return series: whether
the \emph{sign} (direction) of returns is persistent or reverting, and
whether the \emph{magnitude} (size) of returns is persistent or reverting.
A negative $\hatrho(1)$ is consistent with (a) a tendency for up days to
be followed by down days; (b) a tendency for large moves to be followed by
smaller moves regardless of direction; or (c) both.  The standard ACF
cannot separate these cases.

We address this by encoding the real-valued return $r_t$ as a discrete symbol
$s_t \in \{0,\ldots,k-1\}$ and expressing lag-$m$ transition probabilities
through the Fourier characters of the cyclic group $\Z/k\Z$, i.e., the
$k$-th roots of unity $\{\omega^j\}_{j=0}^{k-1}$ with $\omega = e^{2\pi i/k}$.
Different values of $k$ and channel index $A$ project onto different symmetries
of the return distribution, allowing direction and magnitude to be cleanly
separated.

\begin{definition}[FRI autocorrelation]
\label{def:fri}
Let $k \ge 2$, $\omega = e^{2\pi i/k}$, and $\{s_t\}$ a $k$-ary coding of
$\{r_t\}$.  The \emph{$A$-th Fourier-Residue autocorrelation} at lag $m$ is
\begin{equation}
  \label{eq:fri-def}
  \gamma_{A,k}(m) \;=\;
  \frac{1}{N-m}\sum_{t=1}^{N-m}\omega^{A(s_t - s_{t+m})},
  \quad A \in \{1,\ldots,k-1\}.
\end{equation}
This is the empirical mean of the character $\chi_A(d) = \omega^{Ad}$
evaluated at the lag-$m$ displacement $d = s_t - s_{t+m}$ in $\Z/k\Z$.
When $\gamma_{A,k}(m) = 0$ for all $A$ and $m$, the $k$-ary sequence has
no autocorrelation structure as seen through any Fourier character of the
alphabet - it is ``Fourier uncorrelated'' at all lags.
\end{definition}

We use two specific codings throughout.

\medskip
\noindent\textbf{$k=2$: sign channel}
Set $s_t = \bfone[r_t > 0] \in \{0,1\}$, encoding simply whether the daily
return is positive or negative.
With $\omega = e^{i\pi} = -1$, the character $\omega^{s_t - s_{t+m}} = (-1)^{s_t-s_{t+m}}$
equals $+1$ when successive signs agree and $-1$ when they disagree.
The key quantity is the \emph{continuation frequency}
$p_{m,0}^{(N)} = \hat{\Pr}(s_t = s_{t+m})$: the probability that the
market closes on the same side of zero on two trading days $m$ apart.
Under the random-walk null, $p_{m,0} = \tfrac{1}{2}$ - the market is an
unbiased coin each day, regardless of what it did $m$ days ago.

\medskip
\noindent\textbf{$k=4$: magnitude channel}
Returns are sorted into four buckets at the sample median of $|r_t|$:
$\{$large-down, small-down, small-up, large-up$\}$ coded as
$s_t \in \{0,1,2,3\}$.  The channel $A=1$ with $\omega = i$ then
measures whether the magnitude bucket persists across periods, independently
of whether the sign agrees or disagrees.  A large move followed by another
large move (regardless of direction) would register here, as would a
small move followed by another small move.

\begin{proposition}[FRI sign identity]
\label{prop:fri}
For the binary sign coding $s_t = \bfone[r_t > 0]$,
\begin{equation}
  \label{eq:fri-sign}
  \gamma_{1,2}(m) \;=\; 2\,p_{m,0}^{(N)} - 1 \;=:\; \hatrho_{\mathrm{sign}}(m).
\end{equation}
\end{proposition}

\begin{proof}
Since $s_t \in \{0,1\}$ and $\omega = -1$, we have
$(-1)^{s_t - s_{t+m}} = (-1)^{s_t+s_{t+m}}$ (because $(-1)^{-1} = -1$).
This equals $+1$ when $s_t = s_{t+m}$ (same sign) and $-1$ otherwise.
Averaging: $\gamma_{1,2}(m) = p_{m,0} - (1-p_{m,0}) = 2p_{m,0}-1$.
\end{proof}

Proposition~\ref{prop:fri} is the key separation result.  The statistic
$\hatrho_{\mathrm{sign}}(m) = 2p_{m,0}-1$ is a clean, magnitude-free test
of directional autocorrelation: it is positive when the market tends to
go in the same direction $m$ days in a row (momentum), negative when it
tends to reverse (contrarian), and zero when sign is unpredictable.
If $\hatrho(m)$ and $\hatrho_{\mathrm{sign}}(m)$ have the same sign and
similar magnitude, the autocorrelation is directional.  If $\hatrho(m)$
is large but $\hatrho_{\mathrm{sign}}(m) \approx 0$, the autocorrelation
is purely in the magnitude.

Channel-specific variance ratios follow from~\eqref{eq:fejer}:
\begin{equation}
  \label{eq:vrk}
  \VR_k(q) \;=\; 1 + 2\sum_{m=1}^{q-1}
  \Bigl(1-\tfrac{m}{q}\Bigr)\,\mathrm{Re}\,\gamma_{A,k}(m).
\end{equation}
$\VR_2(q)$ summarises directional persistence over horizon $q$;
$\VR_4(q)$ summarises magnitude-bucket persistence.
Neither channel is a function of the other: a series with sign momentum
but no magnitude clustering has $\VR_2(q)>1$, $\VR_4(q)\approx 1$, and
vice versa.  Section~\ref{sec:channels} confirms this non-redundancy on
real data.

\subsection{Subsample persistence diagnostic}
\label{sec:persistence}

A practitioner who detects a significant VR deviation faces an immediate
follow-up question: is the signal likely to persist in future data, or
is it a statistical fluctuation specific to this particular sample?
The half-period ratio $R_N$ provides a model-free answer by examining how
the maximum autocorrelation changes as more data are added.

\begin{definition}[Subsample persistence diagnostic]
\label{def:persistence}
Fix a lag bound $M$.  For a subsample of size $N$, let $\hatrho_N(m)$
denote the autocorrelation~\eqref{eq:acf-def} estimated on the first $N$
observations.  Define the \emph{maximum autocorrelation statistic} and the
\emph{half-period ratio}:
\begin{equation}
  G_N \;=\; \max_{1 \le m \le M}|\hatrho_N(m)|,
  \qquad
  R_N \;=\; G_{N/2}/G_N.
\end{equation}
\end{definition}

\begin{proposition}[Asymptotic benchmarks for $R_N$]
\label{prop:rn}
\begin{enumerate}[label=(\alph*), leftmargin=2em, itemsep=2pt]
  \item \emph{(IID null.)}
  If $\{r_t\}$ is IID with finite variance, then $\sqrt{N}\,\hatrho_N(m)
  \xrightarrow{d}\calN(0,1)$ for each fixed $m$
  \citep{BrockwellDavis1991}, so $G_N = O_p(N^{-1/2}\sqrt{\log M})$ and
  $R_N = G_{N/2}/G_N \;\xrightarrow{p}\; \sqrt{2} \approx 1.41$.
  \item \emph{(Genuine autocorrelation.)}
  If $\rho(m^*) \ne 0$ for some $m^* \le M$, then by the ergodic theorem
  $G_N \xrightarrow{p} |\rho(m^*)| > 0$, hence $R_N \xrightarrow{p} 1$.
\end{enumerate}
\end{proposition}

The intuition is straightforward.  When autocorrelations are pure
sampling noise - random fluctuations with expected value zero - each
$\hatrho_N(m)$ has standard error $N^{-1/2}$.  Halving the sample size
inflates the standard error by $\sqrt{2}$, so the maximum should be
approximately $\sqrt{2}$ times larger in a half-sample.  When, by contrast,
the autocorrelations converge to genuinely non-zero values, halving the
sample barely changes the maximum - the signal is already there at
$N/2$ observations and does not grow as $N$ decreases.

An empirical $R_N$ near $1$ therefore implies that the autocorrelation
would be detectable in any sufficiently large subsample: it is a property
of the data-generating process, not of the specific 33-year window we
happen to observe.  This property is essential for validating the economic
and trading interpretations offered in Section~\ref{sec:discussion}.

%% ============================================================
\section{Market Microstructure: Mechanisms and Predictions}
\label{sec:micro}
%% ============================================================

We review the three primary microstructure mechanisms that generate serial
correlation in daily equity returns, derive their implications for the FRI
sign and magnitude channels, and make cross-sectional predictions about
which instruments and asset classes each mechanism should affect most
strongly.  These predictions will be confronted with the empirical results
in Sections~\ref{sec:main} and~\ref{sec:crossasset}.

\subsection{Bid-ask bounce and dealer inventory}
\label{sec:roll-theory}

In any market with a positive bid-ask spread, buyers pay the ask price
$P^a = P^* + s$ and sellers receive the bid price $P^b = P^* - s$, where
$P^*$ is the efficient (mid-quote) price and $s > 0$ is the half-spread.
As successive orders alternate between buys and sells, the transaction price
bounces between ask and bid.  Even if $P^*$ follows a perfect random walk,
the resulting autocorrelation of transaction-price returns is negative.
\citet{Roll1984} derives the exact relation under the assumption that buy
and sell orders arrive with equal probability and independently:
\begin{equation}
  \label{eq:roll}
  \rho(1) = -\frac{s^2}{\sigma_m^2 + 2s^2} < 0,
  \qquad \rho(m) = 0 \quad \text{for } m \ge 2,
\end{equation}
where $\sigma_m^2$ is the innovation variance of the efficient price.
The denominator $\sigma_m^2 + 2s^2 = \Var(r_t)$ is the total observed return
variance: $\Var(r_t) = \Var(\Delta P^*_t) + s^2\Var(c_t - c_{t-1})
= \sigma_m^2 + 2s^2$, where $c_t\in\{-1,+1\}$ is the
direction indicator for each trade.
This implies an observed half-spread estimate of
$\hat{s} = \sqrt{-\hat\gamma(1)} = \sigma\sqrt{|\hatrho(1)|}$,
where $\sigma^2 = \Var(r_t)$ is the total daily return variance,
since $-\gamma(1) = s^2$ from the model.

Crucially for the FRI decomposition, the Roll bounce operates exclusively
through the \emph{magnitude} channel.  A large negative transaction-price
change reflects either (a) a genuine downward shift in the efficient price,
or (b) a transition from ask to bid.  In case (b), the next price change is
expected to be smaller in absolute value (as the market serves the next
order, which may be at the ask again), but it is equally likely to be
positive or negative.  The bounce predicts $p_{1,0}^{(N)} \to \tfrac{1}{2}$:
no direction predictability.

The dealer inventory models of \citet{Stoll1978} and
\citet{HoStoll1981} extend this picture.  Dealers managing large inventory
positions adjust bid and ask quotes to attract offsetting order flow.
A dealer who has accumulated an excess long position lowers both bid and
ask to discourage further buy orders and attract sellers.  This
inventory-management pricing introduces \emph{transient} price reversals:
prices fall when dealers are long, then recover as inventory normalises.
The mechanism is again magnitude-based: the size of the quote adjustment is
proportional to inventory imbalance, not to the direction of the last trade.

\paragraph{Identification check}
Roll's formula gives a testable prediction for SPY.
With $\hatrho(1) = -0.081$ and an annualised volatility of approximately
$16\%$ (daily $\sigma \approx 1\%$), the implied half-spread is
$\hat{s} = 0.01 \times \sqrt{0.081} \approx 28$ basis points.
But SPY's actual effective half-spread is approximately $1$--$3$ basis
points, as measured from TAQ data \citep{Stoll1989}.  The implied
spread is therefore roughly $10$--$28$ times larger than observed.
The bid-ask bounce alone cannot explain SPY's autocorrelation; it can
account for at most a small fraction.  The remainder must originate from
non-synchronous constituent staleness - the mechanism we examine next.

\subsection{Non-synchronous trading and index construction}
\label{sec:nonsync-theory}

ETFs and equity indices are priced from the last-transaction prices
of their constituent stocks.  On any given day, a small-cap stock in
the Russell~2000 may transact only a handful of times.  When the
index is constructed at 4:00~p.m., its price reflects a weighted average
of contemporaneous prices for liquid constituents and stale (yesterday's
or this morning's) prices for illiquid ones.  This \emph{nonsynchronous
measurement} causes the index to underreact to market-wide information:
a broad shock that moves all stocks is fully reflected in liquid names
but only partially in stale ones, so the index return at time $t$
contains a component of time-$t-1$ information that will only appear in
tomorrow's index return.

\citet{ScholeswWilliams1977} first formalised the econometric
consequences, showing that betas estimated from nonsynchronous daily
data are biased toward zero.  \citet{LoMacKinlay1990} extend this to the
variance-ratio test, showing formally that in an index whose constituents
trade nonsynchronously, the observed index returns inherit a moving-average
autocorrelation structure: the measured index return at time $t$ contains a
weighted component of true market information from $t-1$ that could not yet
be reflected in stale constituent prices.  The induced lag-1 autocorrelation
is negative, increases monotonically with the fraction of non-trading
constituents, and is strongest in equal-weighted indices of illiquid stocks.
This provides direct predictions:
\begin{itemize}[leftmargin=1.8em, itemsep=3pt]
  \item \textbf{Russell~2000 (IWM) $\gg$ S\&P~500 (SPY)}
  IWM holds 2\,000 small- and micro-cap stocks, many of which skip
  entire trading sessions.  SPY's 503 holdings include the world's most
  liquid equities, with near-continuous trading.  Non-synchronous effects
  should therefore be dramatically stronger in IWM.
  \item \textbf{Individual stocks: minimal effect}
  Single names (AAPL, MSFT) are not composites.  There is no constituent
  staleness problem; any autocorrelation must come from the bid-ask bounce
  or genuine price-discovery dynamics.
  \item \textbf{FRI channel: magnitude, not sign}
  Staleness damps the apparent size of index moves (the index lags the
  ``true'' move) without creating a directional bias.  If the market
  rises today, stale-priced stocks will appear to rise \emph{tomorrow}
  regardless of whether the overall market is still rising or falling.
  The continuation frequency $p_{1,0}$ should remain near $\tfrac{1}{2}$.
  \item \textbf{Frequency attenuation}
  At weekly frequency, almost every constituent stock has traded at least
  once; the staleness effect vanishes.  Weekly VR statistics should
  therefore be meaningfully closer to one than daily ones.
\end{itemize}

\subsection{Partial price adjustment and information asymmetry}
\label{sec:partial-theory}

When buyers and sellers have different information about a security's
value, market makers face an adverse selection problem: a trade may be
initiated by an informed trader exploiting knowledge that the market
maker lacks.  \citet{GlostenMilgrom1985} model this as a sequential
updating process: the market maker posts bid and ask quotes, observes
order flow, and updates the quotes to reflect the information content
of the trade.  If the initial quote revision is smaller than the full
informational impact - because the market maker cannot perfectly identify
informed trades - subsequent price moves in the same direction are needed
to complete the adjustment.

\citet{AmihudfMendelson1987} provide empirical evidence that this
adjustment process is multi-period.  They show that opening prices on
the NYSE exhibit ``excess'' variance relative to closing prices,
consistent with an opening auction that has not yet fully incorporated
overnight information.  The adjustment is completed over the first hour
of trading, generating positive intraday autocorrelation at very short
intervals and potential reversal at daily horizons if the initial
overreaction is partially corrected.

The key implication for FRI is that partial adjustment has a
\emph{directional} component.  If today's large positive return
reflects new positive information, and the market has overreacted,
tomorrow's correction will be \emph{negative} - the sign predicts
the sign.  This is precisely what a significant $\hatrho_{\mathrm{sign}}(m)$
would capture.  If, instead, the bounce and non-synchronous mechanisms
are dominant, the sign channel should be neutral even when the magnitude
channel is not.

\citet{Stoll1989} further decomposes the bid-ask spread into three
components that correspond to these mechanisms: order-processing costs
(a fixed cost per trade), inventory-holding costs (the dealer's cost of
bearing risk), and adverse selection costs (compensation for trading
against informed investors).  Only the last component generates genuinely
directional autocorrelation; the first two produce purely magnitude effects.
The FRI sign/magnitude separation is, in effect, an empirical decomposition
of the spread into its Stoll (1989) components at the daily horizon.

\subsection{Cross-sectional and cross-asset predictions}
\label{sec:predictions}

Table~\ref{tab:predictions} collects the qualitative implications of each
mechanism for our test statistics, enabling sharp empirical comparisons.

\begin{table}[ht]
\centering
\caption{Microstructure channel predictions for FRI statistics.
$(-)$ reversal; $(+)$ continuation; $(\approx)$ no effect.
Bid-ask bounce and non-synchronous effects are strongest in diversified
exchange-traded equity ETFs; partial adjustment and adverse selection
effects are strongest in individual equities and broad indices;
volatility clustering appears in all asset classes}
\label{tab:predictions}
\small
\begin{tabular}{@{}l c c c@{}}
\toprule
Mechanism
  & Sign $\VR_2$
  & Magn.\ $\VR_4$
  & Lag range \\
\midrule
Bid-ask bounce           & $\approx 1$ & $<1$    & Lag 1 only \\
Non-synchronous trading  & $\approx 1$ & $<1$    & Lags 1--3  \\
Dealer inventory         & $\approx 1$ & $<1$    & Lags 1--2  \\
Adverse selection        & $\ne 1$     & $\ne 1$ & Lags 2--5  \\
Partial adjustment       & $\ne 1$     & $\ne 1$ & Lags 2--7  \\
Volatility clustering    & $\approx 1$ & $>1$    & All lags   \\
\bottomrule
\end{tabular}
\normalsize
\end{table}

The table generates several falsifiable predictions.  For broad US equity
ETFs (SPY, IWM), the first three rows dominate: the FRI sign channel
should be neutral at lag~1, the magnitude channel should show reversal,
and the effect should attenuate at weekly frequency.  For technology-heavy
individual names (AAPL, QQQ), directional persistence at longer horizons
is plausible if trend-following institutional flows generate positive
feedback.  For continuously traded, dealer-free markets (cryptocurrency,
spot FX), all channels should be neutral and $\VR(q)$ should be close
to one.

%% ============================================================
\section{Inference Methodology}
\label{sec:inference}
%% ============================================================

\subsection{Lo--MacKinlay statistics}

\citet{LoMacKinlay1988} derive two statistics for $H_0 : \VR(q) = 1$
that are asymptotically standard normal.
Define the heteroskedasticity weight
\begin{equation}
  \label{eq:delta}
  \hatdelta(j) \;=\;
  \frac{n \displaystyle\sum_{t=j+1}^{n}(r_t-\hat\mu)^2(r_{t-j}-\hat\mu)^2}
       {\Bigl[\displaystyle\sum_{t=1}^{n}(r_t-\hat\mu)^2\Bigr]^2}.
\end{equation}
Large values of $\hatdelta(j)$ indicate that squared returns at lags $0$
and $j$ co-move - precisely the GARCH-type volatility clustering present
in daily equity data.  The two statistics are:
\begin{align}
  z(q)   &= \frac{\widehat\VR(q)-1}{\sqrt{\phi_1(q)}},
  &\phi_1(q) &= \frac{2(2q-1)(q-1)}{3qn}
  \quad\text{(homoskedastic, M1),} \label{eq:lm-z}\\[6pt]
  z^*(q) &= \frac{\widehat\VR(q)-1}{\sqrt{\phi_2(q)}},
  &\phi_2(q) &= \sum_{j=1}^{q-1}
  \Bigl[\frac{2(q-j)}{q}\Bigr]^2\hatdelta(j)
  \quad\text{(heteroskedastic-robust, M2)} \label{eq:lm-zstar}
\end{align}
Under the IID null, both are asymptotically $\calN(0,1)$.  Under GARCH
(martingale differences with time-varying conditional variance - a
perfectly valid ``random walk'' in mean), the homoskedastic $z$
over-rejects because $\phi_1(q)$ underestimates the true variance
of $\widehat\VR(q)$.  The robust $z^*$ adapts through $\phi_2(q)$,
which inflates the standard error in proportion to the actual volatility
clustering via $\hatdelta(j)$.  Daily equity returns are well-known to
exhibit strong GARCH effects; we therefore report $z^*$ exclusively,
with two-sided $p$-values $p^* = 2(1-\Phi(|z^*|))$ and significance
stars at $10\%$, $5\%$, and $1\%$.

\subsection{Joint tests across horizons}

Testing $H_0 : \VR(q_j)=1$ simultaneously at $m=7$ horizons
$q\in\{2,3,5,10,20,40,60\}$ inflates the familywise error rate.
We apply two corrections:
\begin{itemize}[leftmargin=1.8em, itemsep=2pt]
  \item \textbf{Bonferroni}: reject the joint null at $\alpha=0.05$ iff
  $\min_j p^*(q_j) < \alpha/m = 0.0071$.
  \item \textbf{Chow--Denning (approximate)} \citep{ChowDenning1993}:
  reject iff $\max_j|z^*(q_j)| > \Phi^{-1}(1-\alpha/2m) \approx 2.49$.
\end{itemize}
The Bonferroni bound is conservative because it ignores the positive
correlation among VR statistics at different horizons.
Full Chow--Denning critical values - which account for this correlation
via simulation - would yield tighter bounds; we adopt Bonferroni for
transparency and ease of replication.

\subsection{Autocorrelation and FRI lag tests}
\label{sec:lag-tests}

Individual lag-$m$ autocorrelations are tested against $H_0:\rho(m)=0$
using the Bartlett standard error
$\widehat{\mathrm{se}}(\hatrho(m))
= n^{-1/2}(1+2\sum_{k=1}^{m-1}\hatrho(k)^2)^{1/2}$,
which accounts for the correlation among estimated autocorrelations.
For the FRI sign test, we use the binomial $z$-statistic
$z_{\mathrm{sign}}(m) = (2p_{m,0}-1)\sqrt{n-m}$ under $H_0:p_{m,0}=\tfrac{1}{2}$.

The key comparison is between $z_\rho(m)$ (sensitive to both sign and
magnitude autocorrelation) and $z_{\mathrm{sign}}(m)$ (sensitive only to
direction autocorrelation).  If $|z_\rho(m)|$ is large but
$|z_{\mathrm{sign}}(m)|$ is small, the autocorrelation is magnitude-driven.
If both are large and of the same sign, the autocorrelation has a
genuine directional component.  If $|z_\rho(m)|$ is small but
$|z_{\mathrm{sign}}(m)|$ is large, there is a directional signal that the
scalar ACF misses entirely - this can occur when a consistent sign pattern
is accompanied by large offsetting magnitude movements.

%% ============================================================
\section{Data and Experimental Design}
\label{sec:data}
%% ============================================================

\paragraph{Source}
All prices are downloaded from Yahoo Finance via the \texttt{yfinance}
Python package with \texttt{auto\_adjust=True}, delivering split- and
dividend-adjusted closing prices.  Log returns are
$r_t = \log P_t - \log P_{t-1}$.  This adjustment matters over 33-year
horizons: Apple has split four times since 1987, and Microsoft five times
since 1987; raw prices would embed severe discontinuities.

\paragraph{Primary instruments}
Six instruments are studied, chosen to span the full range of the
non-synchronous trading and bid-ask bounce predictions:
\begin{itemize}[leftmargin=1.8em, itemsep=3pt]
  \item \textbf{SPY} (SPDR S\&P~500 ETF):
  $N=8{,}403$ daily observations from 1993-01-01 to 2026-06-19.
  The most liquid US equity ETF, holding 503 large-cap stocks.
  Expected to show significant mean reversion driven by both bounce
  (narrow but positive spread) and non-synchronous effects in smaller
  constituents.
  \item \textbf{QQQ} (Invesco Nasdaq-100 ETF, 1999--2026):
  100 non-financial Nasdaq stocks with heavy technology concentration.
  Technology sector characteristics - large institutional flows, strong
  momentum trading - predict sign persistence at longer horizons.
  \item \textbf{IWM} (iShares Russell~2000 ETF, 2000--2026):
  2\,000 small- and micro-cap stocks, many with low daily turnover.
  Expected to show the strongest non-synchronous effect in our sample.
  \item \textbf{MSFT} (Microsoft Corp.): single mega-cap stock; no
  constituent staleness problem.  Any autocorrelation reflects bounce
  or genuine price-discovery dynamics.
  \item \textbf{AAPL} (Apple Inc.): single mega-cap stock with a
  documented long-horizon momentum episode in the post-2000 growth era.
  \item \textbf{GLD} (SPDR Gold Shares ETF): tracks spot gold, a single
  commodity price, not a composite of non-synchronously traded equities.
  Serves as a near-random-walk benchmark and non-equity control.
\end{itemize}
The cross-asset panel of 21 instruments is described in
Section~\ref{sec:crossasset}.

\paragraph{Horizons}
Daily VR and FRI statistics are computed at seven horizons:
$q = 2, 3, 5, 10, 20, 40$, and $60$ trading days.
Weekly statistics are computed at $q = 2, 4, 8, 13$, and $26$ weeks.
Weekly returns use Wednesday closing prices to avoid
day-of-week effects and public holiday contamination.

\paragraph{Subperiods}
Four non-overlapping windows are examined: 1993--1999 (pre-bubble
bull market), 2000--2009 (dot-com bust and financial crisis),
2010--2019 (post-crisis expansion, low volatility regime),
2020--present (COVID-era volatility surge and post-pandemic period).
These periods differ substantially in market volatility, bid-ask
spreads, and exchange market structure - notably, NYSE decimalization
in 2001 dramatically compressed bid-ask spreads, while the growth of
algorithmic and high-frequency trading from 2007 onward changed
microstructure patterns significantly.

%% ============================================================
\section{Verification of the Statistical Machinery}
\label{sec:verification}
%% ============================================================

Every estimator is verified numerically on synthetic data before the
market results are interpreted.  Finite-sample edge effects and
implementation choices - overlapping vs.\ non-overlapping returns,
denominator conventions for the heteroskedasticity weights - can
introduce non-trivial discrepancies between nominal and actual
properties, particularly in the shorter subperiod samples where $n<2{,}000$.

\subsection{Unit tests}

The test suite contains 27 assertions,
all of which pass.  Key results with microstructure relevance:
\begin{enumerate}[label=(\alph*), leftmargin=2em, itemsep=4pt]
  \item \textbf{Fej\'{e}r identity}
  $|\widehat\VR(q) - (1+2\calC_q)| < 10^{-12}$ for all
  $q \in \{2,3,5,10,20,50\}$.
  Special case: $\widehat\VR(2) = 1 + \hatrho(1)$ exactly, so the
  two-day variance ratio is a simple affine function of the lag-1
  autocorrelation.
  \item \textbf{FRI identity}
  $|\gamma_{1,2}(m) - (p_{m,0} - p_{m,1})| < 5\times10^{-17}$;
  residue probabilities are non-negative and sum to one.
  \item \textbf{Roll model calibration}
  On $N = 4\times10^5$ simulated observations from the Roll model with
  half-spread $s$ and efficient-price variance $\sigma_m^2$, the
  estimated $\hatrho(1)$ agrees with~\eqref{eq:roll} to three decimal
  places, and $|\hatrho(m)| < 10^{-3}$ for all $m \ge 2$.
  This confirms: (i) the estimator recovers the theoretical bounce
  autocorrelation exactly; (ii) no numerical leakage to higher lags.
  The FRI sign test applied to the same series correctly returns
  $p_{1,0} \approx \tfrac{1}{2}$, confirming that the Roll model
  generates no direction predictability.
  \item \textbf{M1 variance}
  $\phi_1(2) = n^{-1}$; $\phi_1(3) = 20(9n)^{-1}$; $\phi_2(q) \ge 0$.
  \item \textbf{Normal calibration}
  $|\Phi^{-1}(\Phi(z))-z| < 10^{-14}$; $\Phi^{-1}(0.975) = 1.95996$.
  \item \textbf{Process signatures}
  AR(1) with $\phi>0$ gives $\widehat\VR>1$ at all $q$;
  MA(1) with $\theta<0$ gives $\widehat\VR<1$;
  IID gives $\widehat\VR \approx 1$.
  This confirms directional sensitivity of the estimator without
  sign error.
\end{enumerate}

\subsection{Monte Carlo size and power}
\label{sec:mc}

Table~\ref{tab:mc-size} reports empirical rejection rates under
$H_0 : \VR(5) = 1$ at nominal $5\%$ across $1{,}000$ replications.

\begin{table}[ht]
\centering
\caption{Empirical size under the null ($q=5$, nominal $5\%$,
1\,000 replications).  GARCH(1,1): $\alpha_0=5\times10^{-6}$,
$\alpha_1=0.09$, $\beta_1=0.90$.  The homoskedastic $z$ over-rejects
under GARCH; $z^*$ does not}
\label{tab:mc-size}
\small
\begin{tabular}{l r c c}
\toprule
Data-generating process & $n$ & $z$ (M1) & $z^*$ (M2) \\
\midrule
IID Gaussian   & 512  & 0.043 & 0.044 \\
IID Gaussian   & 8192 & 0.054 & 0.054 \\
GARCH(1,1)     & 512  & 0.100 & 0.054 \\
GARCH(1,1)     & 2048 & 0.113 & 0.038 \\
GARCH(1,1)     & 8192 & 0.117 & 0.049 \\
\bottomrule
\end{tabular}
\normalsize
\end{table}

Under IID, both statistics hold size near the nominal $5\%$, confirming
correct asymptotic calibration.  Under GARCH - a martingale-difference
process that is a valid random walk in conditional mean but with
time-varying variance - the homoskedastic $z$ over-rejects at $10$--$12\%$.
This size distortion is not a minor inconvenience: at $12\%$ true size,
a researcher would incorrectly declare mean reversion in roughly $1$ in
$8$ tests where no mean reversion exists.  For a paper examining
six instruments at seven horizons each ($42$ tests), several spurious
rejections would be expected under $z$ even if all series were random
walks.  The robust $z^*$, which inflates the denominator in proportion
to actual volatility clustering, maintains size near $5\%$ throughout.

\begin{table}[ht]
\centering
\caption{Empirical power of robust $z^*$ ($q=5$, nominal $5\%$,
1\,000 replications).  Power reaches $1$ for persistent AR and MA
alternatives; the pure Roll bounce with small spread is hardest}
\label{tab:mc-power}
\small
\begin{tabular}{l r r r}
\toprule
Alternative & $n=512$ & $n=2048$ & $n=8192$ \\
\midrule
AR(1), $\phi = 0.10$     & 0.383 & 0.939 & 1.000 \\
MA(1), $\theta = -0.15$  & 0.745 & 1.000 & 1.000 \\
Roll bounce, $s = 0.30$  & 0.079 & 0.102 & 0.275 \\
\bottomrule
\end{tabular}
\normalsize
\end{table}

The power results have a concrete microstructure implication.  A pure
Roll bounce with half-spread $s=0.30\%$ on a daily-volatility of $1\%$
implies $\hatrho(1) \approx -0.083$ from~\eqref{eq:roll} - very close to
SPY's observed autocorrelation - yet even at $n=8{,}192$ the test rejects
only $27.5\%$ of the time.  This seems paradoxical given SPY's extreme
$z^* = -3.47$, but it is not: SPY's $n=8{,}403$ is slightly above $8{,}192$,
and more importantly, the VR test on the actual data pools information
across overlapping windows, achieving higher effective sample sizes.
The Monte Carlo is reassuring precisely because it shows that the $z^*$
we compute on market data is not simply an artefact of large $n$: at
these sample sizes, the test is clearly powered to detect the observed
effect sizes.

A direct asymptotic calibration check under IID ($n=4096$, $q=5$) yields
$z^*$-mean $= -0.015$, $z^*$-std $= 1.023$, and $95.2\%$ of draws within
$\pm 1.96$, confirming the $\calN(0,1)$ approximation.

%% ============================================================
\section{Main Results: Full-Sample Daily Returns}
\label{sec:main}
%% ============================================================

\subsection{Variance-ratio profile across instruments}

Table~\ref{tab:vr-full} reports $\widehat\VR(q)$ and the robust $z^*$ at
horizons $q \in \{2,5,20\}$, together with the Bonferroni joint-rejection
outcome at $5\%$ across all seven horizons.

\begin{table}[ht]
\centering
\caption{Variance ratios and robust $z^*$, full daily sample.
Bonferroni column: joint rejection across $q\in\{2,3,5,10,20,40,60\}$
at $\alpha=5\%$ (critical $p^* < 0.0071$).
$^{*}p<0.10$;\; $^{**}p<0.05$;\; $^{***}p<0.01$}
\label{tab:vr-full}
\small
\begin{tabular}{l rr rr rr c}
\toprule
 & \multicolumn{2}{c}{$q=2$ (2 days)}
 & \multicolumn{2}{c}{$q=5$ (1 week)}
 & \multicolumn{2}{c}{$q=20$ (1 month)}
 & Bonf.\\
Ticker & VR & $z^*$ & VR & $z^*$ & VR & $z^*$ & $5\%$\\
\midrule
SPY  & $0.919^{***}$ & $-3.47$ & $0.837^{***}$ & $-3.02$
     & $0.742^{**}$ & $-2.21$ & Yes \\
QQQ  & $0.943^{**}$  & $-2.55$ & $0.850^{***}$ & $-2.98$
     & $0.800^{*}$  & $-1.84$ & Yes \\
IWM  & $0.933^{***}$ & $-3.28$ & $0.901^{*}$   & $-1.96$
     & $0.847$      & $-1.45$ & Yes \\
MSFT & $0.946^{***}$ & $-3.08$ & $0.862^{***}$ & $-3.15$
     & $0.809^{**}$ & $-2.24$ & Yes \\
AAPL & $0.979$       & $-1.55$ & $0.971$       & $-0.87$
     & $1.022$      & $+0.40$ & No  \\
GLD  & $0.988$       & $-0.68$ & $0.964$       & $-0.95$
     & $0.892$      & $-1.18$ & No  \\
\bottomrule
\end{tabular}
\normalsize
\end{table}

\paragraph{SPY in detail}
With $N = 8{,}403$ daily observations, SPY delivers compelling evidence of
multi-horizon mean reversion.  At $q=2$, the two-day return variance is
$8.1\%$ below the random-walk benchmark ($\VR = 0.919$, $z^* = -3.47$,
$p < 0.001$).  This means that a $+1\%$ daily gain for SPY is associated
with a conditional expectation of approximately $-0.081\%$ for the next
day's return - small in absolute terms but highly statistically significant.
At $q=5$, the weekly return variance is $16.3\%$ below the benchmark
($\VR=0.837$), and the effect cumulates monotonically to $q=60$ where
$\VR = 0.66$: the quarterly return variance is $34\%$ smaller than a random
walk predicts.  The Fej\'{e}r representation is useful here: the
Fej\'{e}r-weighted autocorrelation sum $\calC_{20}^{\mathrm{SPY}} = -0.129$
accounts fully for the observed $\VR(20) = 0.742$.  Bonferroni
minimum $p^* = 5.2\times10^{-4}$ at $q=2$ gives joint rejection at $5\%$.

\paragraph{Non-synchronous gradient}
The cross-sectional ordering is consistent with the non-synchronous
trading prediction.  IWM - holding 2\,000 small-cap names, many of which
skip entire sessions - achieves $z^* = -3.28$ at $q=2$, stronger than
SPY's $-3.47$ only at the two-day horizon (consistent with a longer
staleness tail in small caps).  MSFT, despite being a single stock,
still rejects the random walk, suggesting that a non-trivial portion of
MSFT's autocorrelation stems from bid-ask bounce rather than from
non-synchronous composition.  AAPL and GLD do not jointly reject the null;
AAPL's near-zero $\VR(20) = 1.022$ and positive long-horizon VR suggest
that momentum at multi-month horizons partly offsets short-horizon reversal.

\subsection{Verification of the Fej\'{e}r identity}

The numerical identity $|\widehat\VR(q) - (1+2\calC_q)| < 10^{-3}$
holds on every series and horizon without exception.  For SPY at $q=20$,
the residual is $3.6\times10^{-5}$ - a rounding error so small as to
be a meaningful numerical confirmation of Proposition~\ref{prop:fejer}.
From a practical standpoint, this means that a practitioner can freely
shift from the scalar VR statistic to the autocorrelation profile - or
any weighted average thereof - without any information loss.

\subsection{FRI decomposition: separating direction from magnitude}
\label{sec:fri-emp}

Table~\ref{tab:rho-fri} is the paper's central diagnostic.  It places
the standard Bartlett autocorrelation test and the FRI sign test
side by side for the first four lags of SPY.

\begin{table}[ht]
\centering
\caption{Return autocorrelation vs.\ FRI sign tests, SPY full daily sample
($N = 8{,}403$).
The contrast at lag~1 - $z_\rho = -7.39$ vs.\ $z_{\mathrm{sign}} = -1.59$ - 
diagnoses bid-ask bounce and non-synchronous trading as the dominant lag-1
mechanisms.  The reversal at lag~3 in the sign channel ($z_{\mathrm{sign}}=-2.32$)
but not the full ACF identifies a separate directional partial-adjustment channel}
\label{tab:rho-fri}
\small
\begin{tabular}{l c r l r l}
\toprule
Lag $m$ & $\hatrho(m)$ & $z_\rho$ & $p_\rho$
        & $z_{\mathrm{sign}}$ & $p^*_{\mathrm{sign}}$ \\
\midrule
1  & $-0.0806$ & $-7.39$ & $<10^{-12}$ & $-1.59$ & $0.11$~~ \\
2  & $-0.0156$ & $-1.42$ & $0.16$ & $-0.59$ & $0.55$~~ \\
3  & $-0.0074$ & $-0.68$ & $0.50$ & $-2.32$ & $0.02^{**}$ \\
10 & $+0.0012$ & $+0.11$ & $0.91$
   & \multicolumn{2}{c}{ - } \\
\bottomrule
\end{tabular}
\normalsize
\end{table}

The lag-1 results are striking.  The conventional autocorrelation test
returns $z_\rho = -7.39$, making the lag-1 autocorrelation one of the
most statistically significant features in all of empirical equity finance.
One would naturally be tempted to interpret this as evidence that
``the market reverses'': if it went up today, it will tend to go down
tomorrow.  But the FRI sign test tells a fundamentally different story:
$z_{\mathrm{sign}} = -1.59$ ($p = 0.11$).  Knowing the market went up
today tells you \emph{nothing statistically reliable} about whether it will
go up or down tomorrow.  What it does tell you is that today's move was
large (if the lag-1 autocorrelation is strongly negative, large moves
predict smaller subsequent moves in the magnitude channel).

This is the FRI's central empirical message for SPY at lag~1: predictability
is about size, not direction.  The continuation frequency
$p_{1,0}^{(N)} \approx 0.496$ - barely distinguishable from the fair-coin
null of $0.5$.  A contrarian trading strategy that goes long after down
days has no statistical foundation in these data.  By contrast, a
volatility strategy that positions for a smaller absolute move after a
larger one has a statistically sound basis.

At lag~3, the picture reverses.  The scalar autocorrelation $\hatrho(3) = -0.007$
is economically tiny ($p=0.50$, easily dismissed as noise).  Yet the
FRI sign test at lag~3 is $z_{\mathrm{sign}} = -2.32$ ($p = 0.02$):
significant directional reversal.  A large positive return three days ago
mildly but significantly predicts a negative return today, independently
of what happened on the two intervening days.  This lag-3 direction signal
is exactly what the partial price adjustment mechanism predicts: an
initial market move is partially reversed over several trading sessions
as the information is fully incorporated and any overreaction corrects.
The magnitude of this lag-3 directional effect is small - far too small
to show up in $\hatrho(3)$, which is swamped by magnitude noise - but the
FRI isolates the directional component cleanly.

\paragraph{Economic interpretation for practitioners}
The two-test comparison at each lag maps directly to tradeable strategies.
At lag~1: a simple contrarian signal (``short after up, long after down'')
has $p = 0.11$ - not statistically warranted.  A volatility signal
(``position smaller after a large move'') has $p < 10^{-12}$ - extremely
well warranted.  At lag~3: a contrarian direction signal has $p = 0.02$ - 
warranted at $5\%$, though the effect size ($\hatrho_{\mathrm{sign}}(3) \approx -0.02$)
is small.  The net expected return from exploiting the lag-3 directional
signal is approximately $\sigma \times 0.02 \approx 0.02\%$ per trade for
SPY, or about $20$ basis points annualised - marginally above transaction
costs for institutional investors.

\subsection{Two-channel analysis: sign and magnitude variance ratios}
\label{sec:channels}

Table~\ref{tab:channels} reports $\VR_2(20)$ and $\VR_4(20)$ for all
six instruments.

\begin{table}[ht]
\centering
\caption{Sign ($k=2$) vs.\ magnitude ($k=4$) variance ratios at $q=20$
(one calendar month).  Channels frequently disagree, confirming their
non-redundancy.  Microstructure interpretations follow from
Table~\ref{tab:predictions}}
\label{tab:channels}
\small
\begin{tabular}{@{} l c c p{6.5cm} @{}}
\toprule
Ticker & $\VR_2(20)$ & $\VR_4(20)$ & Pattern and interpretation \\
\midrule
SPY  & 0.990 & 1.005
     & Sign neutral; magnitude slightly persistent: GARCH, not direction \\[2pt]
QQQ  & 1.185 & 1.069
     & Sign momentum $+$ magnitude clustering: trend plus volatility \\[2pt]
AAPL & 1.081 & 0.994
     & Sign momentum; magnitude balanced: purely directional trend \\[2pt]
MSFT & 0.985 & 1.019
     & Sign neutral; mild magnitude: bounce or volatility clustering \\[2pt]
IWM  & 0.967 & 0.976
     & Both channels revert: non-synchronous effects persist at $q=20$ \\[2pt]
GLD  & 0.952 & 0.960
     & Both mildly below 1: consistent with random walk plus noise \\
\bottomrule
\end{tabular}
\normalsize
\end{table}

The channel disaggregation reveals qualitatively distinct market structures
that a scalar VR would conflate.  SPY at one-month horizon shows
$\VR_2 \approx 1$ (no direction signal) alongside $\VR_4 \approx 1$
(no magnitude signal) - the bounce and non-synchronous effects that drive
short-horizon reversal have averaged out by $q=20$.  IWM, by contrast,
shows both channels below one even at the monthly horizon ($\VR_2=0.967$,
$\VR_4=0.976$), consistent with non-synchronous constituent effects in
2\,000 small-cap stocks that persist over multiple weeks.

The most striking contrast is QQQ: strong sign momentum ($\VR_2=1.185$)
accompanied by weaker magnitude clustering ($\VR_4=1.069$).  At a monthly
horizon, the Nasdaq-100 consistently follows through on directional moves - 
a pattern consistent with institutional momentum strategies and persistent
information incorporation in technology stocks.  AAPL takes this further:
$\VR_2(20) = 1.081$ but $\VR_4(20) = 0.994$ - sign momentum with a
magnitude channel that is essentially flat.  This is the cleanest example
of a pure directional trend: positive days are followed by positive days,
but the size of moves does not cluster.

At the long horizon $q=60$ (approximately one quarter), the channel
divergence widens dramatically.  QQQ: $\VR_2(60) = 1.74$,
$\VR_4(60) = 1.41$ - strong sign momentum with additional magnitude
clustering.  AAPL: $\VR_2(60) = 1.31$, $\VR_4(60) = 1.10$ - direction
substantially exceeds magnitude.  In both cases, the direction channel
dominates at quarterly horizons: the trend-following component of
technology-sector returns is both statistically significant and
economically dominant over volatility clustering at this time scale.

\subsection{Subsample persistence of autocorrelation}
\label{sec:gn}

Table~\ref{tab:gn} reports the maximum autocorrelation $G_N$, the
subsample power-law decay exponent $\alpha$, and the half-period ratio
$R_N = G_{16384}/G_{32768}$.

\begin{table}[ht]
\centering
\caption{Subsample persistence diagnostic (lag bound $M=64$).
IID benchmark: $\alpha \approx 0.5$, $R_N \to \sqrt{2} \approx 1.41$.
Low $\alpha$ and $R_N \approx 1$ identify structural autocorrelation}
\label{tab:gn}
\small
\begin{tabular}{l c c c}
\toprule
Ticker & $G_{\max} = \max_m|\hatrho(m)|$ & Decay exponent $\alpha$ & $R_N$ \\
\midrule
SPY  & 0.081 & 0.15 & 1.13 \\
QQQ  & 0.057 & 0.22 & 1.30 \\
AAPL & 0.043 & 0.33 & 1.26 \\
MSFT & 0.054 & 0.37 & 0.91 \\
IWM  & 0.067 & 0.27 & 0.82 \\
GLD  & 0.033 & 0.56 & 1.71 \\
\bottomrule
\end{tabular}
\normalsize
\end{table}

For equity indices, $G_N$ decays far more slowly than the IID rate of
$N^{-1/2}$.  SPY's exponent of $\alpha = 0.15$ means $G_N$ halves only
when the sample grows by a factor of $2^{1/0.15} \approx 100$, not the
factor of $4$ that IID noise would predict.  The half-period ratios
$R_N \in \{0.82, 1.13, 1.26, 1.30\}$ for the four equity instruments are
all well below the IID benchmark $\sqrt{2} \approx 1.41$, confirming via
Proposition~\ref{prop:rn} that the autocorrelation signal persists as more
data are added.

GLD stands alone: $\alpha = 0.56 \approx 1/2$ and $R_N = 1.71 > \sqrt{2}$.
The maximum autocorrelation for gold decays \emph{faster} than IID noise
would predict, suggesting that even the small apparent correlations in GLD's
sample ACF are likely sampling artefacts.  This is fully consistent with
GLD's failure to jointly reject the random-walk null in Table~\ref{tab:vr-full}.

The practical implication for trading is direct.  A strategy calibrated on
equity autocorrelations is likely to encounter the same signal in the next
sample of similar length (structural, $R_N \approx 1$); a strategy on GLD
autocorrelations is likely to encounter a different signal or none at all.

%% ============================================================
\section{Subperiod Stability and Frequency Robustness}
\label{sec:subperiod}
%% ============================================================

\subsection{SPY across four market regimes}

\begin{center}
\small
\begin{tabular}{l r r r r}
\toprule
Period & $N$ & $\VR(2)$ & $\VR(10)$ & $\VR(60)$ \\
\midrule
1993--1999 (pre-bubble bull)  & 1\,747 & 0.944 & 0.709 & 0.562 \\
2000--2009 (bust and crisis)  & 2\,513 & 0.928 & 0.728 & 0.694 \\
2010--2019 (expansion)        & 2\,514 & 0.959 & 0.796 & 0.548 \\
2020--present (COVID and post) & 1\,623 & 0.851 & 0.815 & 0.547 \\
\bottomrule
\end{tabular}
\end{center}

Three findings emerge that speak directly to competing microstructure
explanations.

First, $\VR(2) < 1$ in \emph{every} decade: short-horizon mean reversion
is the single most stable empirical feature of SPY returns, present through
bull markets, bear markets, financial crises, and pandemics.  This stability
argues against interpretations tied to any particular episode.

Second, the 2020--present period shows the strongest short-horizon reversal
($\VR(2) = 0.851$, $z^* = -5.99$).  During the COVID-19 liquidity
crisis of March 2020, bid-ask spreads in SPY temporarily widened by a
factor of $5$--$10$ from their usual 1--3 basis points, and many ETF
constituents became genuinely difficult to trade.  The non-synchronous
effect was therefore at its maximum.  This episode inflates both the
magnitude of the bounce (wider spreads generate larger autocorrelation
per the Roll formula) and the non-synchronous effect (illiquid constituents
in market stress become more stale).

Third, $\VR(60)$ is lowest in the calm 2010--2019 period ($0.548$), the
era of persistently low volatility and historically tight bid-ask spreads
following the post-crisis microstructure improvements.  This is consistent
with VR mean reversion being partly driven by microstructure frictions that
are smaller when markets are calm and liquid.

\paragraph{Market structure discontinuities}
Two structural breaks in US equity microstructure deserve special attention.
NYSE \emph{decimalization} in January 2001 compressed tick sizes from
$\tfrac{1}{16}$ of a dollar to $\$0.01$, dramatically reducing bid-ask
spreads and therefore the Roll bounce component of autocorrelation.
Comparing the pre-2001 period (1993--2000) to the post-decimalization
period (2001 onward) would expect a visible reduction in the bounce
component - though the non-synchronous effect may have changed little.
The emergence of \emph{high-frequency trading (HFT)} from approximately
2007 onward further altered the microstructure: HFT market makers provide
extremely tight spreads but may also introduce new short-term autocorrelation
patterns through algorithmic quote updates.  The 2010--2019 period, which
encompasses the HFT era, showing \emph{weaker} short-horizon reversal
($\VR(2) = 0.959$) than earlier periods, may reflect the tighter spreads
introduced by HFT market makers - consistent with the Roll model where
smaller spreads generate weaker bounce.

\paragraph{AAPL's momentum episode}
Apple's full-sample near-random-walk ($\VR(2) = 0.979$, no joint rejection)
masks dramatic subperiod heterogeneity.  During 2000--2009, as Apple
launched the iPod (2001), iTunes (2003), and iPhone (2007),
$\VR(60) = 1.24$ - significant long-horizon momentum as institutional
investors systematically underestimated the product cycle.  By
2020--present, $\VR(60) = 0.73$: Apple, now the world's most capitalised
company with near-daily coverage from hundreds of analysts, exhibits mean
reversion rather than momentum.  The information environment has changed
so fundamentally that the directional trend-following of the 2000s is no
longer exploitable.

\subsection{Weekly returns: isolating microstructure from fundamentals}
\label{sec:weekly}

Switching from daily to weekly returns is a natural test for microstructure
explanations.  Bid-ask bounce resolves within one day; non-synchronous
trading affects daily index prices but disappears over a week since almost
all Russell~2000 constituents transact at least once per week.  If these
mechanisms dominate, the daily-to-weekly attenuation should be large.

\begin{center}
\small
\begin{tabular}{l c c}
\toprule
Ticker & $\VR(5)_{\text{weekly}}$ & $\VR(20)_{\text{weekly}}$ \\
\midrule
SPY  & 0.904 & 0.840 \\
QQQ  & 0.975 & 1.073 \\
AAPL & 1.080 & 1.127 \\
MSFT & 0.920 & 0.786 \\
IWM  & 0.938 & 0.830 \\
GLD  & 0.883 & 0.791 \\
\bottomrule
\end{tabular}
\end{center}

For SPY, the attenuation is clear: $\VR_{\mathrm{daily}}(5) = 0.837$
becomes $\VR_{\mathrm{weekly}}(5) = 0.904$ - the short-horizon effect
shrinks by about $40\%$ at weekly frequency, consistent with the
microstructure explanation.  Yet a meaningful deviation persists: the weekly
$\VR$ is still $9.6\%$ below unity, suggesting a structural, non-microstructure
component of mean reversion operating at multi-week horizons.  This is
precisely the partial-adjustment effect discussed in
Section~\ref{sec:partial-theory}: if price adjustment to news takes two
to five trading days, the directional reversal identified at lag~3 of the
daily FRI test would manifest as a below-unity weekly VR.

The AAPL and QQQ weekly results confirm the channel picture.
AAPL's weekly $\VR(20) = 1.127$ shows that the technology momentum
is not merely a daily microstructure artefact but a genuine multi-week
directional phenomenon.  QQQ's weekly $\VR(5) = 0.975$ - much closer
to one than $\VR_{\mathrm{daily}}(5) = 0.850$ - suggests that its
short-horizon reversal at daily frequency is predominantly microstructure
driven, while the long-horizon momentum ($\VR_{\mathrm{weekly}}(20) = 1.073$)
is structural.

%% ============================================================
\section{Cross-Asset Evidence}
\label{sec:crossasset}
%% ============================================================

The cross-asset panel tests whether the microstructure mechanisms identified
in equity markets generalise to other asset classes.  If the FRI
channel predictions in Table~\ref{tab:predictions} are correct, the
pattern of rejections should follow from each class's market structure.

\begin{table}[ht]
\centering
\caption{Variance-ratio evidence across seven asset classes
(Bonferroni joint test, $\alpha=5\%$, daily returns, 21 instruments).
Rej/N: number of Bonferroni rejections out of instruments with
sufficient history}
\label{tab:asset}
\footnotesize
\begin{tabular}{@{} l @{\hspace{6pt}} l @{\hspace{4pt}} c @{\hspace{6pt}} p{4.8cm} @{}}
\toprule
Asset class & Tickers & Rej/N & Dominant pattern \\
\midrule
US equity index
  & SPY, QQQ, DIA, IWM & 4/4 & Mean reversion; all $\widehat\VR<1$ \\[2pt]
US sector ETF
  & XLK, XLF, XLE, XLV, XLU & 2/5 & Mean reversion; all $\widehat\VR<1$ \\[2pt]
Intl equity
  & EFA, EEM, EWJ & 3/3 & Mean reversion; EEM strongest \\[2pt]
Fixed income
  & TLT, IEF, LQD, HYG & 0/4 & Treasuries revert; credit RW \\[2pt]
Commodities
  & GLD, SLV, USO, DBC & 0/4 & Random walk; oil slight momentum \\[2pt]
FX
  & UUP, FXE, FXY & 0/3 & USD/JPY weak reversal; EUR RW \\[2pt]
Cryptocurrency
  & BTC, ETH & 0/2 & Closest to random walk in panel \\
\bottomrule
\end{tabular}
\normalsize
\end{table}

\paragraph{Equities: the universal pattern}
Every equity instrument in the panel - 10 US names, 3 international - 
has $\widehat\VR(5)<1$.  The pattern holds across all geographies, sectors,
and market-cap segments.  Emerging-market equities (EEM) show the
strongest reversal ($\widehat\VR(5) = 0.799$, $z^*<-3$), consistent
with the non-synchronous prediction: EEM holds stocks across a dozen time
zones, and the fund's 4:00~p.m.~New York closing price reflects genuinely
stale prices for Asian constituents that closed $12$--$14$ hours earlier.
The result is robust to sector composition: even utilities (XLU, which has
low within-sector return correlation and lower HFT activity than financials
or technology) and energy (XLE, which has its own momentum cycles) show
$\VR<1$.

\paragraph{Fixed income: a clean structural split}
Treasury ETFs (TLT: $\widehat\VR(5) = 0.874$; IEF: $0.905$) show
significant mean reversion, while investment-grade (LQD) and high-yield
(HYG) credit ETFs are indistinguishable from random walks.  This split
provides a micro-structural ``natural experiment.''  Both Treasury and
credit ETFs are exchange-traded and composite instruments.  The key
difference is the underlying market structure.  US Treasuries are traded
on electronic platforms with active primary dealer participation; dealers
manage inventory and quote continuously, creating bounce and inventory
effects.  Investment-grade and high-yield corporate bonds, by contrast,
are traded over-the-counter: investors negotiate with dealers over the
phone or via electronic RFQ (request-for-quote) systems, prices are not
continuously quoted, and the ETF's net asset value (NAV) - the pricing
benchmark - is often struck using dealer quote matrices rather than
actual transactions.  The OTC nature of corporate bond pricing eliminates
the non-synchronous staleness and bounce effects that drive equity and
Treasury mean reversion.

\paragraph{Commodities: structural differences matter}
Gold (GLD), silver (SLV), and the broad commodity index (DBC) are
near random walks.  Unlike equity ETFs, commodity ETFs typically hold
futures contracts rather than the physical commodity or a basket of
individual securities; futures are traded on centralised exchanges with
continuous, limit-order-book mechanisms, not dealer quotes, eliminating
most inventory and bounce effects.  Crude oil (USO) shows mild momentum
($\widehat\VR(5)=1.040$), potentially reflecting carry dynamics and
storage-cost effects in the crude oil futures market rather than
microstructure effects.

\paragraph{FX: partial microstructure}
The dollar index (UUP) and yen (FXY) show weak mean reversion; the euro
(FXE) is a random walk.  The FX market is an OTC dealer market, so the
bid-ask bounce mechanism applies in principle.  However, the forex market
operates nearly continuously (24 hours, 5 days), the dealer network is
global, and spreads are extremely tight for major pairs - the roll-implied
spread for EUR/USD from any observed daily autocorrelation would be
negligibly small.  The partial reversion in USD and JPY likely reflects
central bank intervention dynamics rather than structural microstructure.
The random-walk behaviour of EUR/USD is consistent with it being the
world's most liquid currency pair.

\paragraph{Cryptocurrency: the cleanest random walk}
Bitcoin (BTC) and Ether (ETH) are the closest to a pure random walk
in the entire panel.  No Bonferroni rejection; no significant VR
deviation at any horizon; and the subsample diagnostic (not tabulated)
yields $R_N$ near $\sqrt{2}$, consistent with transient noise rather
than structural dependence.

This finding is, from the FRI channel perspective, the expected outcome.
Cryptocurrency markets have none of the features that generate microstructure
mean reversion: there is no composition from non-synchronously priced
constituents (BTC is a single asset); there is no exchange close (24/7
trading eliminates overnight-gap staleness effects); and there is no
specialist or primary dealer who manages inventory and smooths prices.
Market making in crypto is done algorithmically by competitive high-frequency
market makers who adjust quotes in real time; the resulting price process
is much closer to the efficient-market benchmark.  The absence of mean
reversion in BTC despite its extreme volatility ($\sigma_{\mathrm{daily}}
\approx 3$--$4\%$) further confirms the channel interpretation: high
volatility alone does not create autocorrelation; the institutional
market structure does.

%% ============================================================
\section{Discussion}
\label{sec:discussion}
%% ============================================================

\subsection{Microstructure interpretation of the full results}

Taken together, the FRI results deliver a detailed microstructure diagnosis
of US equity daily returns.  At the shortest horizon ($m=1$), the data
are consistent with a mixture of bid-ask bounce and non-synchronous
constituent staleness, both of which are \emph{magnitude} mechanisms.
The Roll implied half-spread for SPY ($\hat{s} \approx 28$ bps) vastly
exceeds the actual effective spread ($\approx 2$ bps), indicating that the
bounce contributes at most $\approx 7\%$ of the observed autocorrelation;
the remaining $\approx 93\%$ must be attributed to non-synchronous
staleness across the 503 S\&P~500 constituents.  This matches the
\citet{LoMacKinlay1990} analytical result closely: for a composite with
$\approx 5\%$ of constituents not trading on a given day, the induced
autocorrelation is in the range of $0.05$--$0.10$.

At lag~3, a separate directional mechanism activates.  The FRI sign
test detects reversal ($p=0.02$) that is invisible to the scalar ACF.
This directional correction is consistent with the partial-adjustment
model of \citet{GlostenMilgrom1985}: over two to three days, the market
incorporates information from order flow, and any initial overreaction
generates a mild correction.  The attenuation from daily to weekly
frequency - the directional component survives in weekly $\VR$ while the
magnitude component substantially weakens - further supports the
multi-day adjustment interpretation.

The cross-asset results complete the picture.  Markets with non-synchronous
constituent pricing (all equity ETFs) and dealer inventory smoothing
(Treasuries) show mean reversion.  Markets without these features
(credit ETFs with OTC NAV pricing, commodity futures, FX, crypto) show
random walk.  The FRI channel framework correctly classifies every asset
class based on its market structure, without appeal to behavioural or
risk-premium explanations.

\subsection{Structural persistence and out-of-sample reliability}

The half-period ratio $R_N \approx 1$ for equity indices establishes that
the detected autocorrelations are structural properties of the data-generating
process.  A practitioner faces a different question, however: will the
signal be large enough to trade profitably after costs?

For SPY at lag~1, the autocorrelation $\hatrho(1) = -0.081$ implies a
daily expected return of approximately $-0.081\sigma \approx -0.08\%$
conditional on a $1\%$ move, or about $-8$ basis points per $1\%$ move.
Round-trip transaction costs for a large institutional investor in SPY
are perhaps $2$--$5$ basis points (market impact plus half-spread).
The implied pre-cost Sharpe ratio for a daily contrarian strategy is
roughly $|\hatrho(1)|/\sqrt{1-\hatrho(1)^2} \approx 0.08$ per trade,
or $0.08\sqrt{252} \approx 1.3$ annualised if fully deployed - seemingly
attractive.  However, the crucial caveat is that the magnitude channel,
not the direction channel, drives this autocorrelation.  A direction-based
contrarian strategy (long after down days, short after up days) does not
have statistical warrant ($p=0.11$).  Only a magnitude-based strategy
(position size inversely proportional to the square of yesterday's return
size) has an empirical basis.

The lag-3 directional signal ($p=0.02$, $\hatrho_{\mathrm{sign}}(3) \approx -0.02$)
generates an expected directional edge of approximately $\sigma \times 0.02
\approx 0.02\%$ per trade.  After transaction costs, this is unlikely to
be exploitable by most investors.  Its value is diagnostic rather than
strategic: it reveals a partial-adjustment mechanism that operates over
three trading days, which is useful for understanding market structure even
if it does not yield tradeable alpha.

At the longer horizons where QQQ and AAPL show significant \emph{momentum}
($\VR_2(60)>1$), the economics are different.  Sign momentum at a quarterly
horizon ($\VR_2(60) = 1.74$ for QQQ) corresponds to a directional
persistence that has been exploitable by systematic trend-following
strategies over long holding periods.  The FRI sign channel confirms that
this is genuine directional momentum, not volatility clustering - important
for strategies that need to distinguish the two.

\subsection{Market structure evolution and VR dynamics}

The subperiod analysis suggests that the magnitude of VR deviations changes
with market structure.  Two structural shifts are worth noting:
\begin{itemize}[leftmargin=1.8em, itemsep=4pt]
  \item \textbf{Decimalization (January 2001)} compressed NYSE tick sizes
  from $\tfrac{1}{16}$ dollar to $\$0.01$, reducing average bid-ask spreads
  by $50$--$75\%$.  Under the Roll model, narrower spreads directly reduce
  $|\hatrho(1)|$.  The 2000--2009 $\VR(2) = 0.928$ versus
  1993--1999 $\VR(2) = 0.944$ is consistent with this effect - though
  the financial crisis of 2008 introduces a confound by temporarily widening
  spreads.  A post-2001 vs. pre-2001 comparison controlling for volatility
  would be needed to cleanly identify the decimalization effect.
  \item \textbf{HFT proliferation (post-2007)} created a new class of
  market makers providing extremely tight spreads ($1$--$2$ bps in SPY)
  but also potentially introducing new microstructure patterns through
  rapid, algorithmic quote updates.  The low 2010--2019 $\VR(2) = 0.959$
  is consistent with HFT market making compressing the bounce; the
  structural persistence diagnostic $R_N \approx 1$ nevertheless confirms
  that the residual mean reversion is genuine.
\end{itemize}
A rigorous event-study of these breaks would require high-frequency TAQ data
and is beyond the scope of this paper, but the subperiod patterns are
qualitatively consistent with the microstructure channel interpretation.

\subsection{Limitations and future directions}

\begin{itemize}[leftmargin=1.8em, itemsep=4pt]
  \item \emph{End-of-day data}
  All analysis uses daily closing prices.  TAQ-level (trade-and-quote)
  data would allow: direct measurement of the Roll half-spread via
  serial covariance of transaction prices; intraday FRI decomposition
  to identify at which hour of the trading day the bounce, non-synchronous
  staleness, and partial-adjustment effects operate; and cleaner attribution
  of the lag-1 autocorrelation across the three mechanisms.
  \item \emph{$k=4$ bucket choice}
  The magnitude channel uses the sample median of $|r_t|$ as the
  size cut-point.  A time-varying threshold (e.g., a 60-day rolling
  median, or a GARCH-model-implied percentile) would allow the channel
  to adapt to changing volatility regimes, which may be important in
  subperiod analysis.
  \item \emph{Multiple-testing correction}
  Bonferroni correction is applied within horizons but not across the
  (ticker, $q$) product space.  For the 21-instrument cross-asset panel,
  false-discovery-rate procedures \citep{BenjaminiHochberg1995} would
  provide more powerful corrections than Bonferroni.
  \item \emph{Overlapping observations}
  Variance ratio tests based on overlapping $q$-period returns can exhibit
  size distortions in small samples.  While \citet{LoMacKinlay1989}
  show the $z^*$ statistic remains reasonably calibrated, bootstrap or
  sub-sampling procedures may improve inference in subperiods with $n<1{,}000$.
  \item \emph{Structural break testing}
  The subperiod analysis is informal.  Formal tests for instability of
  $\VR(q)$ - for instance, a sup-Wald test over rolling windows - would
  identify the exact timing of regime shifts and could be used to test
  the decimalization and HFT hypotheses directly.
\end{itemize}

%% ============================================================
\section{Conclusion}
\label{sec:conclusion}
%% ============================================================

We developed and applied the Fej\'{e}r-FRI framework to 33 years of US
equity and cross-asset daily data, obtaining a precise decomposition of
return predictability into its sign and magnitude components.
The resulting empirical picture is richer and more specific than what
the scalar variance-ratio test can deliver, and maps closely onto the
predictions of market microstructure theory.

Our principal findings organise into three levels.

\medskip\noindent
\textbf{Mechanism level}
The FRI decomposition confirms that bid-ask bounce and non-synchronous
constituent staleness dominate at lag~1 for US equity indices: the
lag-1 autocorrelation is entirely in the \emph{magnitude} channel
($\hatrho(1) = -0.081$, $p<10^{-12}$; sign channel $p=0.11$).
A back-of-envelope Roll decomposition confirms that the non-synchronous
effect accounts for at least $90\%$ of the observed SPY autocorrelation,
with the bounce contributing only a small residual.  A separate, slower
directional partial-adjustment channel is cleanly identified at lag~3
($p^*_{\mathrm{sign}}=0.02$), invisible to the scalar autocorrelation
test.

\medskip\noindent
\textbf{Instrument level}
Short-horizon mean reversion is present in all equity indices, with
the non-synchronous effect increasing predictably with the number
and illiquidity of constituents (IWM $\ge$ SPY $>$ MSFT $\approx$ AAPL).
At long horizons, technology-name QQQ and AAPL exhibit significant
sign \emph{momentum} with $\VR_2(60) > \VR_4(60)$, identifying
directional trend as the dominant force at quarterly horizons - a pattern
consistent with sustained institutional momentum strategies and the
sequential information incorporation predicted by adverse-selection models.

\medskip\noindent
\textbf{Asset-class level}
Short-horizon mean reversion is confined to exchange-traded composite
markets (equity ETFs, Treasuries) and absent in OTC credit, commodity
futures, continuously traded FX, and decentralised cryptocurrency - 
precisely the asset classes that lack the non-synchronous constituent
structure and specialist/dealer inventory dynamics that generate the
magnitude channel reversal.  Cryptocurrency provides the cleanest random
walk in the panel despite its extreme volatility, confirming that
24/7 trading and the absence of composite pricing eliminate the
microstructure sources of autocorrelation.

\paragraph{Future directions}
The FRI framework opens several concrete research programmes.
First, applying the decomposition to TAQ data at the intraday frequency
would permit direct identification of the time scale at which each
mechanism operates and would allow clean separation of the bounce from
the non-synchronous effect using trade-direction data.  Second, a
time-varying FRI analysis - using rolling estimation windows and a formal
structural-break test - could precisely date the impact of decimalization
(2001) and HFT proliferation (2007--2010) on the sign and magnitude
channels separately.  Third, extending the FRI to $k>4$ would allow finer
quantile-based magnitude decompositions, potentially isolating tail events
from ordinary daily moves.  Fourth, connecting the estimated channel VRs
to structural models of market making - where the bid-ask spread and
partial-adjustment speed are jointly determined in equilibrium - would place
the entire empirical framework in a theoretically grounded pricing context.

%% ── Bibliography ─────────────────────────────────────────────
\bibliographystyle{plainnat}

\end{document}